\documentclass[twoside,journal]{IEEEtran}
\usepackage{amsfonts}
\usepackage{amssymb}
\usepackage{amsmath}
\usepackage{mathrsfs}
\usepackage{verbatim}
\usepackage{subfigure}
\usepackage{balance}
\usepackage{booktabs}
\usepackage{fancybox}
\usepackage{bm}
\usepackage{extarrows}
\usepackage{algorithm}
\usepackage{algorithmic}
\usepackage{multirow}
\usepackage{array}
\usepackage{graphicx}
\usepackage{epstopdf}
\newtheorem{theorem}{Theorem}
\newtheorem{lemma}{Lemma}
\newtheorem{corollary}{Corollary}
\newtheorem{proof}{Proof}
\newtheorem{proposition}{Proposition}

\begin{document}

\title{\LARGE Optimal Power Allocation for Artificial Noise under Imperfect CSI against Spatially Random Eavesdroppers}
\author{Tong-Xing Zheng,~\IEEEmembership{Student~Member,~IEEE},~and
        Hui-Ming Wang,~\IEEEmembership{Member,~IEEE}

\thanks{
\copyright 2015 IEEE. Personal use of this material is permitted. However, permission to use this material for any other purposes must be obtained from the IEEE by sending a request to pubs-permissions@ieee.org.
}
\thanks{This work was partially supported by the Foundation for the Author of National Excellent Doctoral Dissertation of China under Grant 201340, the National High-Tech Research and Development Program of China under Grant No. 2015AA011306, the New Century Excellent Talents Support Fund of China under Grant NCET-13-0458, the Fok Ying Tong Education Foundation under Grant 141063, the Fundamental Research Funds for the Central University under Grant No. 2013jdgz11, and
the Young Talent Support Fund of Science and Technology of Shaanxi Province under Grant 2015KJXX-01.
The review of this paper was coordinated by Prof. G. Mao.
}
\thanks{
The authors are with the School of Electronics and Information Engineering, and also with the MOE Key Lab for Intelligent Networks and Network Security, Xi¡¯an Jiaotong University, Xi'an, 710049, Shaanxi, China (e-mail: txzheng@stu.xjtu.edu.cn; xjbswhm@gmail.com). H.-M. Wang is the corresponding author.
}}

\maketitle
\vspace{-0.8 cm}

\begin{abstract}
    In this correspondence, we study the secure multi-antenna transmission with artificial noise (AN) under imperfect channel state information in the presence of spatially randomly distributed eavesdroppers.
    We derive the optimal solutions of the power allocation between the information signal and the AN for
    minimizing the secrecy outage probability (SOP) under a target secrecy rate and for maximizing the secrecy rate under a SOP constraint, respectively.
    Moreover, we provide an interesting insight that channel estimation error affects the optimal power allocation strategy in opposite ways for the above two objectives.
     When the estimation error increases, more power should be allocated
     to the information signal if we aim to decrease the rate-constrained SOP,
    whereas more power should be allocated
    to the AN if we aim to increase the SOP-constrained secrecy rate.
\end{abstract}


\IEEEpeerreviewmaketitle

\section{Introduction}

Physical layer security (PLS), which achieves secure transmissions by exploiting the randomness of wireless channels, has drawn considerable attention recently \cite{Yang2015Safeguarding}, \cite{WangMaga15}.
It has been shown that we are able to greatly improve PLS using multi-antenna techniques with global channel state information (CSI).
However, to acquire the CSI of an eavesdropper is very difficult in real wiretap scenarios, since the eavesdropper is usually passive.
Without the eavesdropper's CSI, Goel \emph{et al.} \cite{Goel2008Guaranteeing} proposed a so-called artificial noise (AN) aided multi-antenna transmission strategy, in which the transmitter masked the information-bearing signal by injecting isotropic AN into the null space of the main channel (from the transmitter to a legitimate receiver), thus creating non-decodable interference to potential eavesdroppers while without impairing the legitimate receiver.
This seminal work has unleashed a wave of innovation \cite{Zhang2013Design}-\cite{WangTSP}, and
the AN scheme has become a promising approach to safeguarding wireless communications.

In practice, the CSI of the main channel is acquired by training, channel estimation and feedback, which inevitably result in CSI imperfection.
Some endeavors have studied the AN scheme allowing for imperfect CSI.
For example,
robust beamforming schemes have been  proposed in \cite{Mukherjee2011Robust} for MIMO systems and in \cite{WangTVT} for cooperative relay systems. The effects of channel quantized feedback to the AN scheme are discussed in \cite{LinTWC} and \cite{ Zhang2015Artificial}, while in \cite{WangTSP}, training and feedback have been jointly investigated and optimized.

However, all the aforementioned works ignored the uncertainty of eavesdroppers' spatial positions.
Generally, eavesdroppers are geographically distributed randomly, especially in large-scale wireless networks.
Analyzing secrecy performance in such random wiretap scenarios is fundamentally different from that with deterministic eavesdroppers's locations.

Recently, stochastic geometry theory has provided a powerful tool to analyze network performance by modeling nodes' positions according to some spatial distributions such as a Poisson point process (PPP) \cite{Haenggi2009Stochastic}; it facilitates the study of the AN scheme against random eavesdroppers \cite{Ghogho2011Physical}-\cite{Zheng2015Multi}.
However, the impact of imperfect CSI on designing the AN is still an open problem.
Particularly, it is yet unknown what the optimal power allocation strategy is, and how a channel estimation error influences power allocation and secrecy performance.
Due to the complicated/implicit forms of the objective functions caursed by location randomness and CSI imperfection, previous works can only obtain the optimal power allocation either by exhaustive search or by numerical calculation instead of providing a tractable expression.
This makes it challenging to reveal an explicit analytical relationship between the optimal power allocation and the channel estimation error.
Our research are motivated by the above observations and challenges.

In this correspondence, we study an AN-aided multi-input single-output (MISO) secure transmission against randomly located eavesdroppers under imperfect channel estimation.
We investigate two important performance metrics, namely,
secrecy outage probability (SOP) and secrecy rate, respectively.
The SOP reflects the quality difference between the main and wiretap channels;
the secrecy rate measures the rate efficiency of secure transmission.
We provide the optimal power allocation strategies for the following optimization problems:
\begin{enumerate}
\item Minimizing the SOP subject to a secrecy rate constraint;

\item Maximizing the secrecy rate subject to a SOP constraint.
\end{enumerate}
Furthermore, we draw an interesting conclusion that channel estimation error influences the optimal power allocation in \emph{opposite ways} for the above two objectives.
 \emph{When the estimation error increases, more power should be allocated to the information signal if we aim to decrease the rate-constrained SOP, whereas more power should be given to the AN if we aim to increase the SOP-constrained secrecy rate.}

To the best of our knowledge, we are the first to reveal an explicit analytical relationship between the optimal power allocation and channel estimation error through strict mathematical proofs.
Although existing works have also shown that AN should be exploited to increase the secrecy rate under imperfect CSI in point-to-point transmissions, their conclusions are just extracted from simulations under specific parameter settings, which may not apply to more general cases.

\emph{Notations}:
$(\cdot)^{\dag}$, $(\cdot)^{\mathrm{T}}$, $|\cdot|$, $\|\cdot\|$ denote conjugate, transpose, absolute value, and Euclidean norm, respectively.
$\mathcal{CN}$ denotes the circularly symmetric complex Gaussian distribution with zero mean and unit variance.
$\mathbb{C}^{m\times n}$ denotes the $m\times n$ complex number domain.

\section{System Model and Problem Description}
Consider a secure transmission from a transmitter (Alice) to a legitimate receiver (Bob) overheard by randomly located eavesdroppers (Eves)\footnote{This may correspond to such a scenario that a multi-antenna transmitter Alice provides specific service to a specified subscriber Bob, while the service should be kept secret to eavesdroppers (also named unauthorized users).}.
Alice has $N$ antennas, Bob and Eves each has a single antenna.
Without loss of generality, we place Alice at the origin and Bob at a deterministic position with a distance $r_B$ from Alice.
The locations of Eves are modeled as a homogeneous PPP $\Phi_E$ of density $\lambda_E$ on a 2-D plane with the $k$-th Eve a distance $r_k$ from Alice.

All wireless channels are assume to undergo flat Rayleigh fading together with a large-scale path loss governed by the exponent $\alpha>2$.
The channel vector of a node with a distance $r$ from Alice is characterized as $\bm{h}r^{-\frac{\alpha}{2}}$, where $\bm{h}\in\mathbb{C}^{N\times 1}$ denotes the small-scale fading vector, with independent and identically distributed (i.i.d.) entries $h_{i}\thicksim \mathcal{CN}$.

We focus on a \emph{frequency-division duplex} (FDD) system in which the channel reciprocity no longer holds.
 We assume Bob estimates the main channel with estimation errors, and sends the estimated channel to Alice via an ideal feedback link (e.g., a high-quality link with negligible quantization error).
 In this case, the exact main channel $\bm{h}_b$ can be modeled as
\begin{equation}\label{channel_model}
    \bm{h}_b = \sqrt{1-\tau^2}\hat{\bm{h}}_{b} + \tau\tilde{\bm{h}}_{b},
\end{equation}
where $\hat{\bm{h}}_{b}$ and $\tilde{\bm{h}}_{b}$ denote the estimated channel and estimation error with i.i.d. entries $\hat h_{b,i}, \tilde h_{b,i}\sim\mathcal{CN}(0,1)$. This assumption arises from employing the minimum mean square error (MMSE) estimation\footnote{
The Gaussian error model is a stochastic uncertainty model.
Another widely used model is the deterministic bounded error model, which is more convenient for analyzing the quantized CSI \cite{Zhang2015Artificial}.} \cite{Mukherjee2011Robust}, \cite{Geraci2014Physical}.
Here,
$\tau\in[0,1]$ denotes the error coefficient; $\tau=0$ corresponds to a perfect channel estimation, and $\tau=1$ means no CSI is acquired at all.
For each eavesdropper, although its CSI is unknown, we assume its channel statistics information is available, which is a general assumption when dealing with PLS \cite{Zhang2013Design}-\cite{Zheng2015Multi}.

Recalling the AN scheme in \cite{Goel2008Guaranteeing}, the transmitted signal vector $\bm{x}$ at Alice is designed in the form of
\begin{equation}\label{x}
  \bm{x}=\sqrt{\xi P}\bm{w}s+\sqrt{(1-\xi) P/(N-1)}\bm{Gv},
\end{equation}
where $s$ is the information signal with $\mathbb{E}[|s|^2]=1$, $\bm{v}\in\mathbb{C}^{(N-1)\times 1}$ is an AN vector with i.i.d. entries $v_i\sim\mathcal{CN}$,
and $\xi$ is the power allocation ratio (PAR) of the desired signal power to the total power $P$.
$\bm{w}\triangleq{\hat{\bm{h}}_{b}^{\dag}}/{\|\hat{\bm{h}}_{b}\|}$ is the beamforming vector for the information signal,
$\bm{G}\in\mathbb{C}^{N\times(N-1)}$ is a weighting matrix for the AN.
The columns of $\bm{W}\triangleq[\bm{w} ~\bm{G}]$ constitute an orthogonal basis.
Let $\bm{s}\triangleq\left[s~ \bm{v}^{\mathrm{T}}\right]$, and the received signals at Bob and the $k$-th Eve are given from \eqref{x}
\begin{align}
\label{signal_Bob}
    y_B  &=\sqrt{1-\tau^2}\|\hat{\bm{h}}_{b}\|^2r_B^{-\frac{\alpha}{2}}s
    +\underbrace{\tau\tilde{\bm{h}}_{b}\bm{W} \bm{s}^{\mathrm{T}}
    r_B^{-\frac{\alpha}{2}} + n_B}_{n^{o}_B},\\
\label{signal_Evek}
    y_k&=\bm{h}_{e,k}\bm{w}r_k^{-\frac{\alpha}{2}}s
        + \bm{h}_{e,k}\bm{Gv}r_k^{-\frac{\alpha}{2}} + n_k,\ \forall k\in\Phi_E,
\end{align}
where $\bm{h}_{e,k}$ denotes the channel from Alice to the $k$-th Eve, and $n_B^{o}$ combines the residual channel estimation error and thermal noise.
Without loss of generality, we assume $n_B, n_{k\in\Phi_E}\thicksim \mathcal{CN}$.
The exact capacity expression of the main channel under imperfect receiver CSI is still unavailable.
A commonly used approach is to examine a capacity lower bound by
treating $n_B^{o}$ as the worst-case Gaussian noise\footnote{
The tightness of this capacity lower bound was verified for Gaussian inputs with MMSE channel estimation in \cite{Hassibi2003How,Zhou2009Design}.}.
  By doing so, the SINRs of Bob and the $k$-th Eve are respectively given by
\begin{align}
\label{sinr_Bob}
    \gamma_B &= \xi\kappa(\tau),\\
\label{sinr_Evek}
    \gamma_k &=\frac{\xi P|\bm{h}_{e,k}^{\mathrm{T}}\bm{w}|^2r_k^{-\alpha}
    }{(1-\xi)P\|\bm{h}_{e,k}^{\mathrm{T}} \bm{G}\|^2 r_k^{-\alpha}/(N-1)+1},
\end{align}
where $\kappa(\tau)=
\frac{(1-\tau^2)P\gamma}
    {\tau^2 P+r_B^{\alpha}}$ with $\gamma\triangleq \|\hat{\bm{h}}_{b}\|^2$.
    Eq. \eqref{sinr_Evek} holds for the pessimistic assumption that the $k$-th Eve has perfect knowledge of both $\hat{\bm{h}}_{b}$ and $\tilde{\bm{h}}_{b}$.
    Given that $\tau\in[0,1]$, $\kappa(\tau)$ is a monotonically decreasing function of $\tau$;
    it reflects the accuracy of channel estimation.
Specifically, a small value of $\kappa(\tau)$ corresponds to a low estimation accuracy and vice versa.
Hereafter, we omit $\tau$ from $\kappa(\tau)$ for notational brevity.

We consider the wiretap scenario in which each Eve individually decodes a secret message.
This corresponds to a \emph{compound} wiretap channel model \cite{Liang2009Compound}, and the capacities of the main channel and the equivalent wiretap channel are $C_B = \log_2(1+\gamma_B)$ and $C_E = \log_2(1+\gamma_E)$ with $\gamma_E \triangleq \max_{k\in\Phi_E}\gamma_k$.
Note that the capacity of Eves is determined by the maximum capacity among all links connecting Alice with Eves.
As done in \cite{Zhang2013Design} and \cite{Zheng2015Multi}, after encoding secret information, Alice transmits the codewords and embedded secret messages at rates $C_B$ and $R_S$, respectively.
If \emph{at least one Eve decodes the secret messages}, i.e., $C_E$ exceeds the \emph{rate $ C_B-R_S$ of redundant information} (to protect from eavesdropping), perfect secrecy is compromised and a secrecy outage occurs;
the corresponding SOP is defined as
\begin{equation}\label{sop_def}
  \mathcal{O}\triangleq \mathbb{P}\{C_E>C_B-R_S\}, ~\forall ~C_B> R_S.
\end{equation}

In the following, we will optimize the PAR to minimize the SOP under a target secrecy rate, and to maximize the secrecy rate under a SOP constraint $\mathcal{O}\le\epsilon\in(0,1)$, respectively.
We emphasize that different from existing research with deterministic Eves' positions, the analysis and design here is much more complicated due to the extra spatial randomness.

\section{Secrecy Outage Probability Minimization}
In this section, we optimize the PAR that minimizes the SOP under a target secrecy rate.
Recalling \eqref{sop_def}, Alice transmits only when $C_B=\log_2(1+\xi\kappa)> R_S$, i.e., $\xi>\frac{2^{R_S}-1}{\kappa}$ should hold to guarantee a reliable connection between Alice and Bob.
For ease of notation, throughout the paper we define $T\triangleq 2^{R_S}$, $\omega\triangleq \frac{T-1}{\kappa}$, $\delta\triangleq \frac{2}{\alpha}$, $\beta \triangleq\pi \Gamma\left(1+\delta\right)$, $\theta\triangleq\frac{T-1}{T}$, and $\varphi\triangleq\frac{\xi^{-1}-1}{N-1}$.

The problem of minimizing $\mathcal{O}$ in \eqref{sop_def} is formulated as \begin{equation}\label{min_sop}
  \min_{\xi}~\mathcal{O},\qquad \mathrm{s.t.} \quad \omega<\xi\leq 1.
\end{equation}
Before proceeding to this optimization problem, we provide a closed-form expression of $\mathcal{O}$ over the PPP network.
\begin{lemma}\label{lemma_sop}
    \textit{If $\xi>\omega$, the SOP defined in \eqref{sop_def} is given by
    \vspace{-0.0cm}\begin{equation}\label{pso}
      \mathcal{O}=1-\exp\left(-\beta\lambda_E
      \left({P\theta^{-1}}\right)^{\delta}
      \mathcal{J}(\xi)\right),
    \end{equation}
where $\mathcal{J}(\xi)= \left({\omega}^{-1}-{\xi}^{-1}\right)^{-\delta}
\left(1+\left({\xi}{\omega}^{-1}-1\right)\theta\varphi\right)^{1-N}$.}
\end{lemma}
\begin{proof}
     Substitute $C_B$ and $C_E$ along with \eqref{sinr_Bob} and \eqref{sinr_Evek} into \eqref{sop_def}, and after some algebraic manipulations, we obtain $\mathcal{O}=1-\mathcal{F}_{\gamma_E}
\left(x\right)$ with $x\triangleq\frac{1+\kappa\xi}{T}-1$, where $\mathcal{F}_{\gamma_E}(x)$ is
     the cumulative distribution function (CDF) of $\gamma_{E}$, which is
\begin{align}\label{cdf_eta_e_app}
\mathcal{F}_{\gamma_E}(x)
    &=\mathbb{P}
    \left\{\max_{k\in\Phi_E}\gamma_k<x\right\}
    =\mathbb{E}_{\Phi_E}\left[\prod_{k\in\Phi_E}
    \mathbb{P}\{\gamma_k<x\}\right]\nonumber\\
  &  \stackrel{\mathrm{(a)}}
    =\mathbb{E}_{\Phi_E}
    \left[\prod_{k\in\Phi_E}\left(
    1-e^{-\frac{r_k^{\alpha}x}{P\xi}}(1+\varphi x)^{1-N} \right)   \right]\nonumber\\
    &\stackrel{\mathrm{(b)}}=
    \exp\left(-2\pi\lambda_E (1+\varphi x)^{1-N}
    \int_0^{\infty}
    e^{-\frac{r^{\alpha}x}{P\xi}}
    rdr\right)\nonumber\\
   & =\exp\left(-\beta\lambda_E(P\xi)^{\delta}
    x^{-\delta}(1+\varphi x)^{1-N}\right),
\end{align}
     where (a) holds for the CDF of $\gamma_k$ \cite{Zhang2013Design}, and (b) holds for the probability generating functional (PGFL) over a PPP \cite{Stoyan1996Stochastic}.
Substituting \eqref{cdf_eta_e_app} into $\mathcal{O}$ completes the proof.
\end{proof}

\begin{figure}[!t]
\centering
\includegraphics[width = 3.0in]{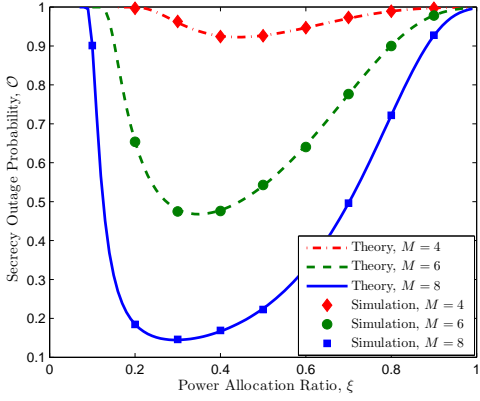}
\caption{SOP $\mathcal{O}$ vs. $\xi$ for different values of $M$, with $P=10$dBm, $R_S=2$, $\tau=0.3$, and $\lambda_E=2$.
Unless otherwise specified, we set $\alpha=4$, $r_B=1$ (unit distance), and the unit of $R_S$ is bits/s/Hz.}
\label{SOP_PHI}
\vspace{-0.0 cm}
\end{figure}

The theoretical values of $\mathcal{O}$ are well verified by Monte-Carlo simulations, as shown in Fig. \ref{SOP_PHI}.
We see that adding transmit antennas is beneficial for decreasing the SOP.
We also observe that as $\xi$ increases, $\mathcal{O}$ first decreases and then increases; there exists a unique $\xi$ that minimizes $\mathcal{O}$.
 In the following we are going to calculate the value of this unique $\xi$.
From \eqref{pso}, it is apparent that minimizing $\mathcal{O}$ is equivalent to minimizing $\mathcal{J}(\xi)$.
The first-order derivative of $\mathcal{J}(\xi)$ on $\xi$ is given by
\begin{equation}\label{dJ1}
  \frac{d\mathcal{J}(\xi)}{d\xi} = \frac{\theta\mathcal{J}(\xi)(\xi^3+a\xi^2+b\xi+c)}
  {\xi^2\left(\xi-\omega\right)(\omega+(\xi-\omega)\theta\varphi)},
\end{equation}
where $a\triangleq-l_1\omega$, $b\triangleq-\frac{\delta}{\theta}\omega^2-l_0\omega^2-l_2\omega$, and $c\triangleq l_2\omega^2$, with $l_0\triangleq\frac{\delta}{N-1}$, $l_1\triangleq 1-l_0$, and $l_2\triangleq 1+l_0$.
Let $\mathcal{K}(\xi)=\xi^3+a\xi^2+b\xi+c$.
 Since $\xi>\omega$, the sign of $\frac{d\mathcal{J}(\xi)}{d\xi}$ follows that of $\mathcal{K}(\xi)$.
In other words, to investigate the monotonicity of $\mathcal{J}(\xi)$ on $\xi$, we need to just examine the sign of $\mathcal{K}(\xi)$.
In the following theorem, we provide the solution to problem \eqref{min_sop}.
\begin{theorem}\label{opt_par_sop_theorem}
\textit{     The optimal PAR that minimizes $\mathcal{O}$ in \eqref{min_sop} is
\begin{align}\label{opt_par_sop}
    \xi^* = \begin{cases}
    ~\varnothing, &0<\kappa\le T-1\\
    ~ 1, &T-1<\kappa\le (T-1)\left(1+\sqrt{{\delta}/{\theta}}\right)\\
    ~\xi_o,&\text{otherwise}
    \end{cases}
\end{align}
    where $\xi_o= \sqrt[3]{q+p}+ \sqrt[3]{q-p} -\frac{a}{3}$ with $p\triangleq \sqrt{\left(\frac{b}{3}-\frac{a^2}{9}\right)^3
    +q^2}$ and $q\triangleq \frac{ab}{6}-\frac{c}{2}-\frac{2a^3}{54}$, and $a$, $b$, $c$ have been defined in \eqref{dJ1}.
    $\xi^* =\varnothing$ means that transmission is suspended. }
\end{theorem}
\begin{proof}
We know that Alice transmits only when $\xi>\omega$.
1)
If $\omega\ge 1$, i.e., $\kappa\le T-1$, no feasible $\xi\in[0,1]$ satisfies $\xi>\omega$, and transmission is suspended.
2)
If $\omega<1$, Alice transmits in the range of $\xi\in(\omega,1]$.
   Next, we derive the optimal value of $\xi$ that minimizes $\mathcal{J}(\xi)$.

We first prove the convexity of $\mathcal{K}(\xi)$ on $\xi\in(\omega,1]$.
From the expression of $\mathcal{K}(\xi)$, we have $\frac{d^2\mathcal{K}(\xi)}{d\xi^2}=6\xi-2l_1\omega>0$, i.e., $\mathcal{K}(\xi)$ is a \emph{convex} function of $\xi$.
Then we determine the sign of $\mathcal{K}(\xi)$.
     The values of $\mathcal{K}(\xi)$ at boundaries $\xi=\omega$ and $\xi=1$ are $\mathcal{K}(\omega)=-\frac{\delta}{\theta}\omega^3$ and $\mathcal{K}(1)=(1-\omega)^2
     -\frac{\delta}{\theta}\omega^2$, respectively.
Obviously, $\mathcal{K}(\omega)<0$ always holds.
Next, we discuss the optimal value of $\xi$ for the following two cases.

\emph{Case 1}: $\mathcal{K}(1)\le0$.
Since $\mathcal{K}(\xi)$ is convex on $\xi\in(\omega, 1]$, $\mathcal{K}(\xi)$ or $\frac{d\mathcal{J}(\xi)}{d\xi}$ is always negative.
Hence $\mathcal{J}(\xi)$ monotonically decreases with $\xi$, and the minimum $\mathcal{J}(\xi)$ is achieved at $\xi=1$, with the corresponding condition obtained from $\mathcal{K}(1)\le 0$, which is $\frac{1}{1+\sqrt{\delta/\theta}}\le\omega<1$.

\emph{Case 2}: $\mathcal{K}(1)>0$.
It means $\mathcal{K}(\xi)$ or $\frac{d\mathcal{J}(\xi)}{d\xi}$ becomes first negative and then positive as $\xi$ increases from $\omega$ to 1, i.e., $\mathcal{J}(\xi)$ first decreases and then increases with $\xi$, and the optimal value of $\xi$ is the unique root of the cubic equation $\mathcal{K}(\xi)=0$.
     Solving this equation using Cardano's formula yields $\xi_o$.

Combining \emph{Case 1} and \emph{Case 2} completes the proof.
\end{proof}

Theorem \ref{opt_par_sop_theorem} indicates that when the value of $\kappa$ is small which corresponds to a poor link quality or a large channel estimation error, Alice either suspends the transmission or transmits with full power.
When the value of $\kappa$ becomes large enough, it is wise to create AN to decrease the SOP.
The resulting minimum SOP, denoted as $\mathcal{O}^*$, is obtained by substituting $\xi^*$ into \eqref{pso}.

Next, we investigate the influence of channel estimation error on the optimal PAR.
Although we obtain a closed-form expression of $\xi_o$ in \eqref{opt_par_sop}, it is complicated to reveal the explicit connection between $\xi_o$ and $\kappa$.
Nevertheless, by leveraging the equation $\mathcal{K}(\xi_o)=0$, we develop some insights into the behavior of $\xi_o$ with respect to $\kappa$ in the following proposition.
\begin{proposition}\label{opt_par_tau_proposition}
    \textit{    $\xi_o$ monotonically decreases with $\kappa$.}
\end{proposition}
\begin{proof}
Since $\omega=\frac{T-1}{\kappa}$, to complete the proof, we need to just prove the monotonicity of $\xi_o$ on $\omega$.
 Utilizing the derivative rule for implicit functions \cite{Jittorntrum1978Implicit} with $\mathcal{K}(\xi_o) = 0$, we obtain
\begin{equation}\label{dxi_domega}
      \frac{d\xi_o}{d\omega}=
\frac{-\partial\mathcal{K}
      /\partial\omega}{\partial \mathcal{K}/\partial\xi_o}=
      \frac{l_1\xi_o^2+{2\frac{\delta}{\theta}\omega}\xi_o
      +2l_0\omega\xi_o+l_2\xi_o-2l_2\omega}
      {3\xi_o^2-2l_1\omega\xi_o-\frac{\delta}{\theta}\omega^2
      -l_2\omega-l_0\omega^2}.
\end{equation}
    Substituting $a$, $b$ and $c$ defined in \eqref{dJ1} into $\mathcal{K}(\xi_o) = 0$ yields
\begin{equation}\label{theta_xi}
    \frac{\delta}{\theta} =\frac{(\xi_o^2+l_0\omega\xi_o-l_2\omega)(\xi_o-\omega)}
      {\omega^2\xi_o}.
\end{equation}
    Since $\xi_o>\omega$ and $\frac{\delta}{\theta}>0$, the term $\xi_o^2+l_0\omega\xi_o-l_2\omega$ in \eqref{theta_xi} satisfies the following inequality
    \vspace{-0.0cm}\begin{equation}\label{theta_xi_inequ}
 0<\xi_o^2+l_0\omega\xi_o-l_2\omega
      <\xi_o^2+l_0\xi_o^2-l_2\omega
<l_2(\xi_0^2-\omega)\Rightarrow \xi_o>\sqrt{\omega}.\nonumber
    \vspace{-0.0cm}\end{equation}
    Substituting $\frac{\delta}{\theta}$ in \eqref{theta_xi} into \eqref{dxi_domega} yields the numerator
    $-\frac{\partial\mathcal{K}}{\partial\omega}=
    \frac{\xi_o}{\omega}[l_1\xi_o(\xi_o-\omega)
      +l_2(\xi_o^2-\omega)]>0$
    and denominator $\frac{\partial \mathcal{K}}{\partial\xi_o}=l_1\xi_o(\xi_o-\omega)
    +\frac{l_2}{\xi_o}(\xi_o^3-\omega^2)>0$, hence we have $\frac{d\xi_o}{d\omega}>0$.
      Combined with
      $\omega=\frac{T-1}{\kappa}$, we directly obtain $\frac{d\xi_o}{d\kappa}=
      \frac{d\xi_o}{d\omega}\frac{d\omega}{d\kappa}<0$, which completes the proof.
\end{proof}

\begin{figure}[!t]
\centering
\includegraphics[width = 3.0in]{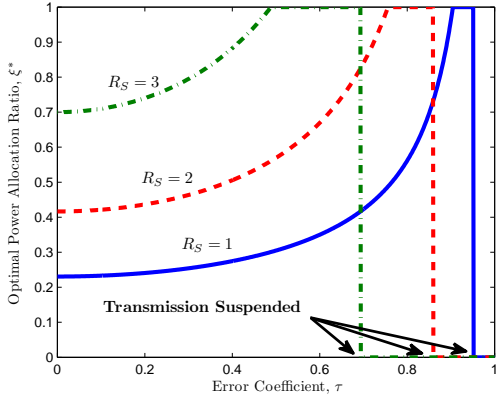}
\caption{Optimal PAR $\xi^*$ vs. error coefficient $\tau$ for different values of $R_S$, with $P=0$dBm, $N=\gamma =20$, and $\lambda_E=2$.}
\label{PAR_SOP_RS}
\vspace{-0.0 cm}
\end{figure}

Proposition \ref{opt_par_tau_proposition} shows that, \emph{when the channel estimation error gets larger, if we aim to decrease the SOP under a target secrecy rate, we should increase the information signal power}, which is validated in Fig. \ref{PAR_SOP_RS}.
It is because that, in order to minimize the SOP, we should first guarantee the link quality of the main channel to support the target secrecy rate.
Hence, we should increase the information signal power to balance the deterioration caused by the channel estimation error.
When $\tau$ exceeds a certain value, transmission is suspended, which is just as analyzed previously.
We also find from Fig. \ref{PAR_SOP_RS} that the value of $\xi^*$ increases as $R_S$ increases, which can be easily confirmed by the fact $\frac{d\xi_o}{dT}=
\frac{d\xi_o}{d\omega}\frac{d\omega}{dT}>0$.

\begin{figure}[!t]
\centering
\includegraphics[width =3.0in]{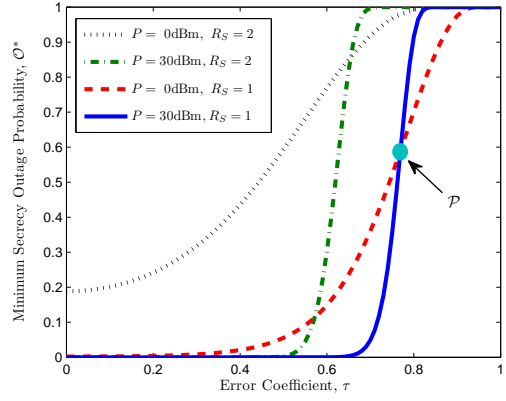}
\caption{Minimum SOP $\mathcal{O}^*$ vs. $\tau$ for different values of $P$ and $R_S$, with $N=\gamma =20$, and $\lambda_E=2$.}
\label{SOP_TAU_P_RS}
\vspace{-0.0 cm}
\end{figure}

Fig. \ref{SOP_TAU_P_RS} shows that the minimum SOP $\mathcal{O}^*$ increases with $\tau$.
For a given $P$, $\mathcal{O}^*$ increases with $R_S$.
For a given $R_S$, the two curves with different values of $P$ cross as $\tau$ increases (see the intersection $\mathcal{P}$).
Specifically, before $\tau$ exceeds $\mathcal{P}$, increasing $P$ decreases $\mathcal{O}^*$, and after that the opposite happens.
This transition occurs because for too large an estimation error, increasing transmit power does not significantly improve Bob's capacity, whereas it is of great benefit to Eves.
This result implies that using full power is not always advantageous, particularly when the estimation error is large. 
\section{Secrecy Rate Maximization}
In this section, we optimize the PAR that maximizes the secrecy rate subject to a SOP constraint.
We first transform the SOP constraint $\mathcal{O}\le\epsilon$ into the following equivalent form
\begin{align}\label{SOP_constraint}
1-\mathcal{F}_{\gamma_E}
  \left(\frac{1+\kappa\xi}{2^{R_S}}-1\right)
  &\stackrel{\mathrm{(c)}}\le\epsilon
\Rightarrow
  \frac{1+\kappa\xi}{2^{R_S}}-1 \geq \mathcal{F}_{\gamma_E}^{-1}(1-\epsilon)\nonumber\\
&  \Rightarrow
  R_S\le\log_2\frac{1+\kappa\xi}{1+\varrho(\xi)\xi},
\end{align}
where (c) holds due to the monotonically increasing feature of the CDF $\mathcal{F}_{\gamma_E}(x)$ on $x$.
$\varrho(\xi)\triangleq
\frac{\mathcal{F}_{\gamma_E}^{-1}(1-\epsilon)}{\xi}$ with $\mathcal{F}_{\gamma_E}^{-1}(\cdot)$ the inverse function of $\mathcal{F}_{\gamma_E}(\cdot)$.
Clearly, a positive value of $R_S$ that satisfies the SOP constraint \eqref{SOP_constraint} exits only when $\varrho(\xi)<\kappa$.
The problem of maximizing $R_S$ can be formulated as
\begin{equation}\label{rs_max}
\max_{\xi}  R_S=\log_2\frac{1+\kappa\xi}
{1+\varrho(\xi)\xi}\quad
\mathrm{s.t.}~ \varrho(\xi)<\kappa,~0\leq \xi\leq 1.
\end{equation}

\begin{figure}[!t]
\centering
\includegraphics[width = 3.0in]{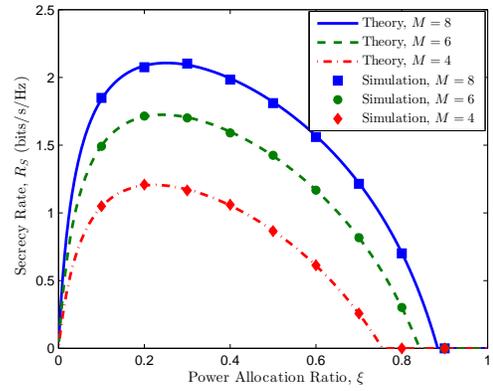}
\caption{Secrecy rate $R_S$ versus $\xi$ for different values of $M$, with $\epsilon=0.01$, $\tau=0.2$, and $\lambda_E=5$.}
\label{RS_PHI}
\end{figure}

An illustration on the relationship between the secrecy rate  and the PAR is shown in Fig. \ref{RS_PHI}.
It is intuitive that increasing the number of antennas helps to improve the secrecy rate.
We observe that, $R_S$ first increases with $\xi$, then decreases with it, and even reduces to zero for too large a $\xi$.
 This implies we should carefully choose the PAR to achieve a high secrecy rate.

From \eqref{rs_max}, we see that the value of $R_S$ is bottlenecked by $\varrho(\xi)$, which implicitly reflects the influence of the density of PPP Eves $\lambda_E$ and the SOP threshold $\epsilon$.
For example, a larger $\lambda_E$ or a smaller $\epsilon$ increases $\varrho(\xi)$ (see \eqref{pso} and the definition of $\varrho(\xi)$), and then decreases $R_S$ (see \eqref{rs_max}). Therefore, $\varrho(\xi)$ plays a critical role in maximizing $R_S$.
Although it is intractable to obtain an analytical expression of $\varrho(\xi)$ due to the transcendental equation $1-\mathcal{F}_{\gamma_E}(\xi\varrho(\xi))=\epsilon$ (see \eqref{cdf_eta_e_app}), we provide an explicit connection between $\varrho(\xi)$ and $\xi$ in the following lemma, which is very critical for the subsequent optimization.
\begin{lemma}\label{varrho_lemma}
    \textit{$\varrho(\xi)$ is a monotonically increasing and convex function of $\xi\in[0,1]$.}
\end{lemma}
\begin{proof}
For notational brevity, we omit $\xi$ from $\varrho(\xi)$.
    Plugging $x=\xi\varrho$ into $1-\mathcal{F}_{\gamma_E}(x)=\epsilon$ yields
\begin{equation}\label{Z}
      \mathcal{Z}(\xi,\varrho) - L = 0,
\end{equation}
    where $\mathcal{Z}(\xi,\varrho)=
    \varrho^{\delta}
    \left(1+\varrho\frac{1-\xi}{N-1}\right)^{N-1}$, and $L\triangleq \frac{\beta\lambda_EP^{\delta}}{-\ln(1-\epsilon)}$. Using the derivative rule for implicit functions with \eqref{Z}, the first- and second-order derivatives of $\varrho$ on $\xi$ are given by
\begin{align}\label{dx1}
      \frac{d\varrho}{d\xi}
      &=-\frac{{\partial \mathcal{Z}}/{\partial\xi}}
      {{\partial \mathcal{Z}}/{\partial \varrho}}
      =\frac{\varrho^2}{\delta+l_2(1-\xi)\varrho},\\
\label{dx2}
  \frac{d^2\varrho}{d\xi^2}
  &=\frac{2}{\varrho}\left(\frac{d\varrho}{d\xi}\right)^2
  +\frac{l_2\varrho^2\left(\varrho-(1-\xi)\frac{d\varrho}
  {d\xi}\right)}
  {\left(\delta+l_2(1-\xi)\varrho\right)^2}.
\end{align}
Clearly, $\frac{d\varrho}{d\xi}>0$ always holds.
With \eqref{dx1}, we have $\varrho-(1-\xi)\frac{d\varrho}{d\xi}=
\frac{\delta\varrho+l_0(1-\xi)\varrho^2}
{\delta+l_2(1-\xi)\varrho}>0$ in \eqref{dx2}.
Removing the second term from the right-hand side of \eqref{dx2} yields $\frac{d^2\varrho}{d\xi^2}
  >\frac{2}{\varrho}\left(\frac{d\varrho}{d\xi}\right)^2>0$.
With $\frac{d\varrho}{d\xi}>0$ and $\frac{d^2\varrho}{d\xi^2}>0$, we complete the proof.
\end{proof}

Lemma \ref{varrho_lemma} indicates the maximum value of $\varrho(\xi)$ is achieved at $\xi=1$, which is $\varrho_{max}=\varrho(1)=L^{{1}/{\delta}}$ from \eqref{Z}.
Besides, it is clearly that $\mathcal{Z}(\xi,\varrho) - L$ monotonically increases with $\varrho$ for a given $\xi$.
Generally, we can calculate the unique value of $\varrho(\xi)$ that satisfies \eqref{Z} using the bisection method in the range $[0,\varrho_{max}]$.
For the special case of large antennas, i.e., $N\rightarrow\infty$, we provide an approximate value of $\varrho(\xi)$, denoted as $\varrho^{o}(\xi)$.
Simulation results show that, when $N\ge 20$, the maximum value of $R_S$ calculated based on $\varrho^{o}(\xi)$ is quite close to that based on the exact $\varrho(\xi)$, i.e., $\varrho^{o}(\xi)$ can be a computationally convenient alternative to $\varrho(\xi)$ when $N$ is large.
\begin{corollary}
\textit{
As $N\rightarrow\infty$, $\varrho(\xi)$ in \eqref{Z} approximates to
\begin{align}\label{varrho}
    \varrho^{o}(\xi) = \begin{cases}
    \qquad L^{{1}/{\delta}}, & \xi = 1\\
    \frac{\delta}{1-\xi}\ln\left(
    \frac{{\delta}^{-1}(1-\xi)L
    ^{{1}/{\delta}}}
    {\mathcal{W}\left({\delta}^{-1}(1-\xi)L
    ^{{1}/{\delta}}\right)}\right), & \text{otherwise}
    \end{cases}
\end{align}
    where $\mathcal{W}(\cdot)$ is the Lambert-W function.}
\end{corollary}
\begin{proof}
Since $\lim_{N\rightarrow\infty}\left(1+\frac{x}{N}\right)^{-N}
= e^{-x}$, we have $\mathcal{Z}(\xi,\varrho^o)=
    (\varrho^o)^{\delta}e^{(1-\xi)\varrho^o}$ where $\varrho^o\triangleq \lim_{N\rightarrow\infty}\varrho$, and \eqref{Z} transforms to $L=
    (\varrho^o)^{\delta}e^{(1-\xi)\varrho^o}$.
    1) When $\xi=1$, we easily obtain $\varrho^o= L^{\frac{1}{\delta}}$.
2) When $\xi\ne 1$, we find that $\frac{1-\xi}{\delta}L^{\frac{1}{\delta}}=
\frac{(1-\xi)\varrho^o}{\delta}
e^{\frac{(1-\xi)\varrho^o}{\delta}}$.
Let $\mu\triangleq \frac{1-\xi}{\delta}\varrho^o$, and we obtain $\frac{1-\xi}{\delta}L^{\frac{1}{\delta}}=\mu e^\mu\Rightarrow e^{\frac{1-\xi}{\delta}L^{\frac{1}{\delta}}}=e^{\mu e^\mu}$.
We further let $\nu\triangleq e^\mu$ and $t\triangleq e^{\frac{1-\xi}{\delta}L^{\frac{1}{\delta}}}$, such that $\nu^\nu = t\Rightarrow\nu=\frac{\ln t}{\mathcal{W}(\ln t)}$.
The solution $\varrho^o$ can be given by $\varrho^o = \frac{1-\xi}{\delta}\mu= \frac{1-\xi}{\delta}\ln \nu$, with yields the final expression in \eqref{varrho} by substituting in $\nu$ along with $t$.
\end{proof}

Due to the implicit function of $\varrho(\xi)$ on $\xi$, we can hardly derive an explicit expression of $R_S$.
Nevertheless, we still reveal the concavity of $R_S$ on $\xi$, and provide the solution to problem \eqref{rs_max} in the following theorem.
\begin{theorem}\label{opt_par_rs_theorem}
\textit{$R_S$ in \eqref{rs_max} is a concave function of $\xi$.
     The optimal $\xi$ that maximizes $R_S$ is given by
\begin{align}\label{opt_par_rs}
        \xi^*=\begin{cases}
        ~\varnothing, & \kappa\le\varrho_{min}\\
        ~1,& \kappa>\frac{\delta L^{{\alpha}/{2}}+L^{\alpha}}
        {\delta-L^{\alpha}}
        ~\textrm{and}~L<\sqrt[\alpha]{\delta}\\
        ~\xi_r,&\textrm{otherwise}
        \end{cases}
\end{align}
    where $\varrho_{min}\triangleq \varrho(0)$ denotes the minimum value of $\varrho(\xi)$.
    $\xi_r$ is the unique root of    $\frac{dR_S}{d\xi}=0$, where
\begin{equation}\label{drs}
  \frac{dR_S}{d\xi}=\frac{1}{\ln2}\left[\frac{\kappa}{1+\kappa\xi}
  -\frac{\varrho(\xi)+\frac{\xi d\varrho(\xi)}{d\xi}}
  {1+\xi\varrho(\xi)}
  \right].
\end{equation}
}
\end{theorem}
\begin{proof}
   Alice transmits only when $\varrho<\kappa$.
Obviously, if $\kappa\le \varrho_{min}$, then $\varrho<\kappa$ never holds for an arbitrary $\varrho$ since $\varrho_{min}\le\varrho$, such that transmission is suspended.
If $\kappa>\varrho_{min}$, Alice transmits for a $\xi$ that satisfies $\varrho<\kappa$.
 To maximize $R_S$, we first give the second-order derivative of $R_S$ on $\xi$ from \eqref{rs_max}
\begin{equation}\label{drs2}
\frac{d^2R_S}{d\xi^2}=\frac{1}{\ln2}
\left[\frac{-\kappa^2}{(1+\kappa\xi)^2}
  -\frac{2\frac{d\varrho}{d\xi}
 +\xi\frac{d^2\varrho}{d\xi^2}}{(1+\varrho\xi)} +\frac{\left(\varrho+\xi\frac{d\varrho}{d\xi}\right)^2}{(1+\varrho\xi)^2}
\right],\nonumber
\end{equation}
with $\frac{d\varrho}{d\xi}$ and $\frac{d^2\varrho}{d\xi^2}$ given in Lemma \ref{varrho_lemma}.
Substituting $\frac{d^2\varrho}{d\xi^2}
  >\frac{2}{\varrho}\left(\frac{d\varrho}{d\xi}\right)^2>0$ (see Lemma \ref{varrho_lemma}) into the above equation yields
\begin{equation}\label{drs22}
  \frac{d^2R_S}{d\xi^2}<-\frac{1}{\ln2}
  \left(\frac{\kappa^2}{(1+\kappa\xi)^2}
  -\frac{\varrho^2}{(1+\varrho\xi)^2}\right).
\end{equation}
 Since $\varrho<\kappa$, we have $\frac{\kappa^2}{(1+\kappa\xi)^2}
  -\frac{\varrho^2}{(1+\varrho\xi)^2}>0
  \Rightarrow\frac{d^2R_S}{d\xi^2}<0$, i.e., $R_S$ is a \emph{concave} function of $\xi$.

Due to the concavity of $R_S$ on $\xi$, the maximum value of $R_S$ is achieved either at boundaries or at stationary points.
From \eqref{drs}, the boundary values are $\frac{dR_S}{d\xi}|_{\xi=0}
=\frac{\kappa-\varrho_{min}}{\ln2}$ and $\frac{dR_S}{d\xi}|_{\xi=1}=\frac{1}{\ln2}
  \left(\frac{\kappa}{1+\kappa}
  -\frac{L^{{\alpha}/{2}}+\frac{\alpha}{2}
  L^{\alpha}}  {1+L^{{\alpha}/{2}}}\right)$.
Obviously, $\frac{dR_S}{d\xi}|_{\xi=0}>0$.
1) If $\frac{dR_S}{d\xi}|_{\xi=1}>0$, $R_S$ monotonically increases with $\xi$, and the optimal value of $\xi$ is 1, with the corresponding condition directly obtained from $\frac{dR_S}{d\xi}|_{\xi=1}>0$.
2) If $\frac{dR_S}{d\xi}|_{\xi=1}\le0$, $R_S$ first increases and then decreases with $\xi$, and the optimal value of $\xi$ is the unique root of $\frac{dRs}{d\xi}=0$.
\end{proof}

Theorem \ref{opt_par_rs_theorem} shows only for a large $\kappa$ (small estimation error) and a small $L$ (a sparse-eavesdropper scenario or a moderate SOP constraint), allocating full power to the information signal provides a higher secrecy rate than the AN scheme does, otherwise generating AN is advantageous.
 Since $R_S$ is a concave function of $\xi$, we can efficiently calculate the unique root $\xi_r$ of $\frac{dR_S}{d\xi_r}=0$ in \eqref{drs} using the bisection method.
Substituting $\xi^*$ and $\varrho(\xi^*)$ into \eqref{rs_max} yields $R_S^*$.

Although $\xi_r$ can only be calculated numerically, we show how $\xi_r$ is affected by $\kappa$ in the following.
\begin{proposition} \label{opt_par_tau_proposition2}
    \textit{$\xi_r$ in \eqref{opt_par_rs} monotonically increases with $\kappa$.
    }
\end{proposition}
\begin{proof}
From \eqref{drs}, $\frac{dR_S}{d\xi_r}=0$ transforms to $\mathcal{A}(\xi_r) = 0$, and
\begin{equation}\label{A}
\mathcal{A}(\xi_r)= (\kappa\xi_r^2-l_0\xi_r+l_2)\varrho_r^2
  +\left(l_2\kappa\xi_r-l_2\kappa+\delta\right)\varrho_r
  -\delta\kappa,
\end{equation}
with $\varrho_r\triangleq \varrho(\xi_r)$.
Using the derivative rule for implicit functions with the equation $\mathcal{A}(\xi_r) = 0$ yields
\begin{equation}\label{dpar_dkappa}
  \frac{d\xi_r}{d\kappa}=-\frac{{\partial \mathcal{A}}/{\partial\kappa}}{{\partial \mathcal{A}}/{\partial\xi_r}}
  =-\frac{\varrho_r^2\xi_r^2-\left(\delta
  +l_2(1-\xi_r)\varrho_r\right)}
 {\psi_1(\xi_r)+\psi_2(\xi_r)\frac{d\varrho_r}{d\xi_r}},
\end{equation}
where $\psi_1(\xi_r)=(1+2\kappa\xi_r)\varrho_r^2
+l_2(\kappa-\varrho_r)\varrho_r$ and $\psi_2(\xi_r)
=2\left(\kappa\xi_r^2+l_2\right)\varrho_r+
\left(l_2\kappa\xi_r+\delta\right)
-2l_0\xi_r\varrho_r-l_2\kappa$.
Obviously, $\kappa>\varrho_r\Rightarrow\psi_1(\xi_r)>0$ and $\frac{d\varrho_r}{d\xi_r}>0$ (see Lemma \ref{varrho_lemma}).
$\mathcal{A}(\xi_r) = 0$ can be further reformed as $\left(\kappa\xi_r^2+l_2)\varrho_r
+(l_2\kappa\xi_r+\delta\right)
=l_0\xi_r\varrho_r+l_2\kappa+\frac{\delta\kappa}{\varrho_r}$, substituting which into $\psi_2(\xi_r)$ directly yields $\psi_2(\xi_r)>0$.
Hence we have $\frac{\partial \mathcal{A}}{\partial\xi_r}>0$.
Leveraging \eqref{drs}, $\frac{dR_S}{d\xi_r}=0$ can be reformed by
$\xi_r^2\frac{d\varrho_r}{d\xi_r}=1-\frac{1+\varrho_r\xi_r}
{1+\kappa\xi_r}<1$.
Substituting $\frac{d\varrho_r}{d\xi_r}$ in \eqref{dx1} into this inequality yields
$\varrho_r^2\xi_r^2<\left(\delta+l_2(1-\xi_r)\varrho_r\right)$, i.e., $\frac{\partial \mathcal{A}}{\partial\kappa}<0$.
  With $\frac{\partial \mathcal{A}}{\partial\xi_r}>0$ and $\frac{\partial \mathcal{A}}{\partial\kappa}<0$, we see from \eqref{dpar_dkappa} that $\frac{d\xi_r}{d\kappa}>0$, which completes the proof.
\end{proof}

Proposition \ref{opt_par_tau_proposition2} indicates that, \emph{when channel estimation error becomes larger, if we aim to increase the secrecy rate under a SOP constraint, we should increase the AN power}, just as shown in Fig. \ref{PAR_RS_EP_LE}.
The reason is: channel estimation error heavily degrades the main channel while has no effect on the wiretap channels.
For a large estimation error, although increasing the information signal power improves Bob's capacity, the improvement is not significant.
On the contrary, increasing AN power always greatly deteriorates the wiretap channels regardless of CSI imperfection.
Therefore, when estimation error becomes larger, increasing AN power is more beneficial to the secrecy rate than increasing signal power.
Nevertheless, transmission is suspended if $\tau$ exceeds a certain value, which corresponds to the case $\kappa\le\varrho_{min}$ as indicated in Theorem \ref{opt_par_rs_theorem}.
We can also prove $\frac{d\xi_r}{d\lambda_E}<0$ and $\frac{d\xi_r}{d\epsilon}>0$ in a similar way as the proof of Proposition \ref{opt_par_tau_proposition2}.
 Due to space limit, we omit the relevant proofs, and the results are verified in Fig. \ref{PAR_RS_EP_LE}.
  We see that the optimal PAR $\xi^*$ decreases for a larger $\lambda_E$ or a smaller $\epsilon$.
 It means that, when transmission is more vulnerable to wiretapping, we should increase AN power.

\begin{figure}[!t]
\centering
\includegraphics[width =3.0in]{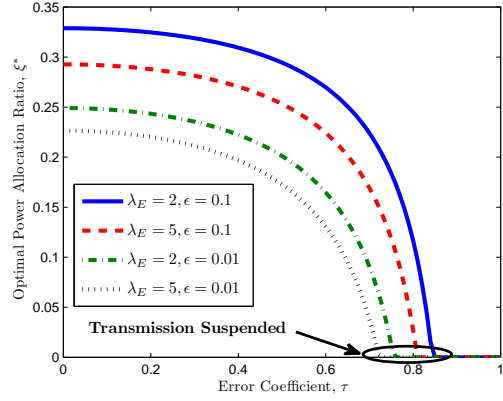}
\caption{Optimal PAR $\xi^*$ versus $\tau$ for different values of $\lambda_E$ and $\epsilon$, with $P=0$dBm, and $N=\gamma = 20$.}
\label{PAR_RS_EP_LE}
\vspace{-0.0 cm}
\end{figure}

\begin{figure}[!t]
\centering
\includegraphics[width = 3.0in]{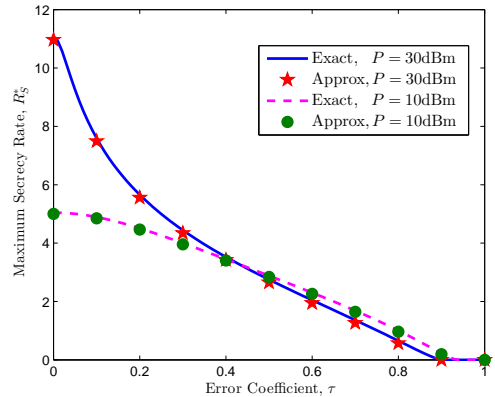}
\caption{Maximum secrecy rate $R_S^*$ versus $\tau$ for different values of $P$, with $N=\gamma = 20$, $\lambda_E=2$, and $\epsilon = 0.01$.
``Approx" corresponds to the value of $\varrho^o(\xi)$ in \eqref{varrho} as opposed to the exact value of $\varrho(\xi)$ obtained from \eqref{Z}.}
\label{RS_TAU_P}
\vspace{-0.2 cm}
\end{figure}

Fig. \ref{RS_TAU_P} depicts the maximum secrecy rate $R_S^*$ versus $\tau$.
The approximated value of $R_S^*$ is quite close to the exact one.
We observe that $R_S^*$ monotonically decreases with $\tau$.
Interestingly, $R_S^*$ increases with $P$ at the small $\tau$ region, whereas decreases with it at the large $\tau$ region.
The underlying reason is just similar to the explanation for the intersection in Fig. \ref{SOP_TAU_P_RS}.

\section{Conclusions}
In this correspondence, we investigate the AN-aided multi-antenna transmission under imperfect CSI against PPP Eves.
We provide explicit solutions of the optimal PARs with channel estimation errors for minimizing the SOP under a secrecy rate constraint and for maximizing the secrecy rate subject to a SOP constraint, respectively.
We strictly prove that, when the channel estimation error becomes larger, we should increase the information signal power if we aim to decrease the SOP, whereas we should increase the AN power if we aim to increase the secrecy rate.

\end{document}